\documentclass[12pt]{article}

\usepackage{microtype}
\usepackage[margin=1in]{geometry}

\usepackage{graphicx}
\usepackage{hyperref}
\usepackage{amsthm}
\usepackage{amssymb}
\usepackage{amsmath}
\usepackage{mathrsfs}
\usepackage{verbatim}
\usepackage{algorithm}
\usepackage{mathtools}
\usepackage{xspace}
\usepackage{enumitem}
\usepackage{color}
\usepackage{url}

\newtheorem{lemma}{Lemma}

\newtheorem{theorem}[lemma]{Theorem}

\newcommand{\ignore}[1]{}

\newcommand{\eps}{\varepsilon}
\newcommand{\reals}{\mathbb{R}}

\newcommand{\dist}{{\tt d}}

\renewcommand{\flat}{F}
\newcommand{\Gflat}{G}

\newcommand{\core}{C}
\renewcommand{\span}[1]{\langle#1\rangle}
\newcommand{\f}{f}
\newcommand{\flatset}{\mathcal{F}}

\newcommand{\Cech}{\v{C}ech\xspace}
\newcommand{\meb}{\mathrm{meb}}
\newcommand{\diam}{\mathrm{diam}}

\newcommand{\cech}{\mathcal{C}}
\newcommand{\proj}{\hat{\pi}}

\newcommand{\Z}{\mathbb{Z}}

\newtheorem{proposition}[lemma]{Proposition}
\newtheorem{corollary}[lemma]{Corollary}

\title{Approximation and Streaming Algorithms for Projective Clustering via Random Projections}

\author{
        Michael Kerber\thanks{Max Planck Institute for Informatics,
  Saarbr\"ucken, Germany, {\tt  mkerber@mpi-inf.mpg.de}}
\and
Sharath Raghvendra\thanks{Virginia Tech, Blacksburg, USA, {\tt sharathr@vt.edu}}
}

\index{Kerber, Michael}
\index{Raghvendra, Sharath}
\begin{document}
\thispagestyle{empty}

\maketitle

\begin{abstract}
Let $P$ be a set of $n$ points in $\reals^d$. In the projective clustering problem, given $k, q$ and norm $\rho \in [1,\infty]$, we have to compute a set $\mathcal{F}$ of $k$ $q$-dimensional flats such that $(\sum_{p\in P}\dist(p, \mathcal{F})^\rho)^{1/\rho}$ is minimized; here $\dist(p, \mathcal{F})$ represents the (Euclidean) distance of $p$ to the closest flat in $\mathcal{F}$. 
We let $\f_k^q(P,\rho)$ denote the minimal value and interpret $\f_k^q(P,\infty)$ to be $\max_{r\in P}\dist(r, \mathcal{F})$. 
When $\rho=1,2$ and $\infty$ and $q=0$, the problem corresponds to the $k$-median, $k$-mean and the $k$-center clustering problems respectively. 

For every $0 < \eps < 1$, $S\subset P$ and $\rho \ge 1$, we show that the orthogonal projection of $P$ onto a randomly chosen flat of dimension $O(((q+1)^2\log(1/\eps)/\eps^3) \log n)$ will $\eps$-approximate $f_1^q(S,\rho)$. 
This result combines the concepts of geometric coresets and subspace embeddings based on the Johnson-Lindenstrauss Lemma.
As a consequence, an orthogonal projection of $P$ to an $O(((q+1)^2 \log ((q+1)/\eps)/\eps^3) \log n)$ dimensional randomly chosen subspace $\eps$-approximates projective clusterings for every $k$ and $\rho$ simultaneously.
Note that the dimension of this subspace is independent of the number of clusters~$k$. 

Using this dimension reduction result, we obtain new approximation and streaming algorithms for projective clustering problems. 
For example, given a stream of $n$ points, we show how to compute an $\eps$-approximate projective clustering for every $k$ and $\rho$ simultaneously using only $O((n+d)((q+1)^2\log ((q+1)/\eps))/\eps^3 \log n)$ space.
Compared to standard streaming algorithms with  $\Omega(kd)$ space requirement, our approach is a significant improvement when the number of input points and their dimensions are of the same order of magnitude.
\end{abstract}

\section{Introduction}
Consider the \emph{projective clustering problem}:
For a set $P$ of $n$ points in $\reals^d$, given integers $k, q < n$ and an integer norm $\rho \ge 1$, compute a set $\mathcal{F}$ of $k$ $q$-dimensional flats (or \emph{$q$-flats}) such that $(\sum_{p\in P}\dist(p, \mathcal{F})^\rho)^{1/\rho}$ is minimized; 
here $\dist(p, \mathcal{F})$ represents the Euclidean distance of $p$ to its closest point on any flat in $\mathcal{F}$. 
We define
\[\f_k^q(P,\rho):=\min_{\mathcal{F}}(\sum_{p\in P}\dist(p, \mathcal{F})^\rho)^{1/\rho}\]
and interpret $\f_k^q(P,\infty)$ to be $\min_{\mathcal{F}}\max_{p\in P}\dist(p, \mathcal{F})$. 
The projective clustering problem is a generalization of several well-known problems. For example, when $\rho=\infty$, $q=0$ this problem is the \emph{minimum enclosing ball} (MEB) problem (when $k=1$) and the \emph{$k$-center clustering problem} (for arbitrary $k$). When $\rho=\infty$ and $q=1$, we get the \emph{minimum enclosing cylinder} (MEC) (for $k=1$) and the \emph{$k$-cylinder clustering problem} (for arbitrary $k$). When $q=0$, we get the $k$-median clustering problem (for $\rho=1$) and the $k$-means clustering problem (for $\rho=2$).
The projective clustering problem is $NP$-Hard~\cite{hardness} and, therefore, most research has focused on the design of approximation algorithms.
For an error parameter $0 < \eps < 1$, an \emph{$\eps$-approximate projective clustering} is a set of $q$-flats $\tilde{\mathcal{F}}$ such that $(\sum_{p\in P}\dist(p, \tilde{\mathcal{F}})^\rho)^{1/\rho}\leq (1+\eps) f_k^q(P,\rho)$.

Projective clustering is an important task arising in unsupervised learning, data mining, computer vision and bioinformatics; see~\cite{proj_survey} for a survey of some of these applications. 
Given its significance, clustering problems have received much attention leading to new approximation algorithms. 
The early algorithms for these problems had exponential dependence on $d$~\cite{am-means, am-means1} 
and were well-suited for low-dimensional inputs. 
However, for many practical problems, the number of input points $n$ and their dimension $d$ are in the same order of magnitude~\cite{fss-turning}.

Badoiu, Indyk and Har-Peled~\cite{bih-cluster} made a breakthrough in the design of high-dimensional clustering algorithms. 
They designed a \emph{coreset}-based algorithm that quickly constructs a small ``most-relevant'' subset $E$ of the input points $P$ 
with the property that an optimal clustering on $E$ is an approximate clustering for $P$,
and use this coreset to compute an approximate clustering.
Based on this idea, several coreset-based approximation algorithms for projective clustering were developed,
also for the design of \emph{streaming} algorithms for projective clustering%
\footnote{In the streaming setting, algorithms are allowed to make one or few passes over the data and compute an approximate solution using a small workspace.}; 
see for example~\cite{chen-coresets,fmsw-coresets,hm-coresets}.
In recent research, depending on the problem, different definitions of coreset have been used. These definitions vary from weak to strong notions of when a subset is relevant, and therefore
yield different size bounds (see for instance~\cite{fss-turning} for a careful discussion).

Throughout this paper, we use the following definition: a coreset (with respect to $\eps$, $q$, $\rho$)
is a subset $E\subseteq P$ such that the affine subspace spanned by $E$ contains a $q$-flat $F$ with 
$(\sum_{p\in P}\dist(p, F)^\rho)^{1/\rho}\leq (1+\eps) f_1^q(P,\rho)$.
We let $C_\rho(q,\eps)$ denote the worst-case size of such a coreset for approximating $\f_1^q(P,\rho)$. 
This is a comparably weak version of coresets: we only require that the subspace spanned by $E$ contains some $\eps$-approximate solution;
we do \emph{not} require that the optimal solution for $E$ is that $\eps$-approximation.
For problems such as MEB, MEC, $1$-mean, or $1$-median, 
there are coresets whose size is independent of the number of points and the ambient dimension~\cite{bc-smaller,bih-cluster,hv-02,sv-efficient}.

Another useful tool for the design of high-dimensional clustering algorithms is the \emph{random projection} method~\cite{vempala-book}.
At its heart is the following well-known
lemma~\cite{jl-extensions} which says that an orthogonal projection of any point set to a random $O(\log n/\eps^2)$-dimensional flat $\eps$-approximates pairwise distance between all pairs of points; see~\cite{dg-elementary} for an elementary proof. 
\begin{theorem}[Johnson-Lindenstrauss]
\label{jl-lemma}
For $0<\eps<1$, a set $P\subset\reals^d$ of $n$ points, and $m\geq 36\ln(n)/\eps^2$,
there is a map $\proj:\reals^d\rightarrow \reals^m$ such that
\[(1-\eps)\|u-v\|^2\leq \|\proj(u)-\proj(v)\|^2\leq (1+\eps)\|u-v\|^2\]
for any $u,v\in P$. Moreover, a randomly chosen map $\proj$ 
of the form $\proj(p)=\sqrt{d/m}\cdot\pi(p)$
where $\pi$ is the orthogonal projection to a $m$-dimensional subspace of $\reals^d$,
satisfies that property with probability at least $1/2$.
\end{theorem}
We abuse notations and refer to $\proj$ as a \emph{random projection} to an $m$-dimensional flat. 

The Johnson-Lindenstrauss Lemma shows that random projections approximate pairwise distances between points. A natural question is what other geometric and structural properties of high-dimensional point cloud are preserved by random projections, and numerous such properties have been identified~\cite{ahh_manifold,c_manifold,in_bounded,magen,sarlos}. 
Random projection techniques are widely used for clustering problems: ongoing research focuses mostly to the case of $k$-means clustering~\cite{randomkmeans,cw-low,cemmp-dimensionality},
although it has also been used for certain projective clustering problems~\cite{bih-cluster, or_clustering}.

\subparagraph{Our results.}
We establish a link between coresets and the random projections for the projective clustering problem in Section~\ref{sec:JL-meta}. 
We show that, for every $0 < \eps < 1$, $q\geq 0$, and $\rho \ge 1$, a random projection to a $O(((q+1)^2\log ((q+1)/\eps)/\eps^3\log n)$-dimensional space $\eps$-approximates $\f_1^q(S,\rho)$ for all $S\subseteq P$. 
The main ingredient of our proof is to show that a random projection to an $O(C_\rho(q,\eps)\log n/\eps^2)$-dimensional subspace ``preserves" all flats defined by subsets of size $C_\rho(q,\eps)$.%
Our argument follows the standard proof technique for subspace embeddings (as sketched in~\cite{cw-low,nn-lower}) by approximately preserving the lengths of vectors taken from a sufficiently dense $\eps$-net.
For a given $q$ and any $\rho\ge 1$, the existence of small-sized coresets with $C_\rho(q,\eps)= O((q+1)^2/\eps \log ((q+1)/\eps))$ is known~\cite{sv-efficient}. 
This leads to the previously mentioned bound on the dimension of the projected space.  

As a consequence, we show that by projecting to the same dimension, also $\f_k^q(P,\rho)$ is preserved for all $k$ and $\rho \ge 1$. 
Note that the dimension of the subspace is independent of $k$ and $\rho$ and is only logarithmic in $n$.
Our results imply that improved bounds on the size of the coreset $C_\rho(q,\eps)$ lead to better bounds on the dimension of the random subspace. 
Interestingly, unlike previous applications of coresets, we do not require a fast method to compute $C_\rho(q,\eps)$. 
Therefore, we can shoot for even smaller-sized coresets without being restricted by its computation time (Section~\ref{sec:coresets}). 

Our results has the following applications (Section~\ref{sec:applications}):
For a given $q$ and a stream of $n$ points, we give an algorithm that can compute projective clustering of $P$ for every value of $k$ and $\rho$ using only $O(((q+1)^2\log ((q+1)/\eps)/\eps^3)(n+d)\log n)$ space. 
Almost all known (multi-pass) streaming algorithms for projective clustering problems have a linear dependence on the product of $k$ and $d$, 
and therefore, they tend to require $\Omega(nd)$ space for when $k = \Theta(n)$. 
As opposed to this, our algorithm requires $\tilde{O}(n+d)$ space which is particularly useful when $n$ and $d$ are of the same order of magnitude. 
Also, in many practical scenarios, the number of clusters $k$ and the norm $\rho$ are not known in advance. Our algorithm is also useful in such cases since our dimension reduction technique works for all values
of $k$ and $\rho$ simultaneously.

We also generically improve approximation algorithms for projective clustering problems.
Again, we project $P$ onto a random subspace and compute an approximate solution in the projected subspace.
We obtain a solution in the original $d$-dimensional space by ``lifting'' each cluster from the projected space separately.
For the approximate $k$-cylinder problem, our approach yields a bound of
$O(n\log n 2^{k\log k/\eps}+\frac{dn\log n}{\eps^3})$ which improves the previously known best
$O(nd2^{k\log k/\eps})$~\cite{bc-smaller}; note that $k$ and $d$ are decoupled in our complexity bound.

Finally, since our results imply that, under random projections, the radius of MEB is approximated for every subset of the input, 
we immediately get an approximation scheme for a $d$-dimensional \Cech complex in Euclidean space by a \Cech complex in lower dimensions.
In particular, this result bounds the persistence of high-dimensional homology classes of the original \Cech complex.
Recently, these results have been proven independently by Sheehy~\cite{sheehy-persistent}.

\section{Generalized Johnson-Lindenstrauss Lemma}
\label{sec:JL-meta}
Recall the definition of $\f_1^q(P,\rho)$ as the $L_\rho$-distance
of $P$ to the best fitting $q$-flat.
We show that a random projection to appropriately large subspaces
approximately preserves $\f_1^q(S,\rho)$ for any subset $S\subseteq P$.
What dimension is appropriate for a projection
depends on the corresponding coreset size $\core:=C_\rho(q,\eps)$;
precisely, picking a $O(\core\log(n)/\eps^2)$-dimensional
subspace is enough.

We outline the proof of the statement before giving the technical details
in the remainder of the section.
For a set $S\subset P$, we let $\span{S}$ denote the \emph{span} of $S$,
that is, the subspace spanned by the points in $S$.
We know that any subset of $P$ has a coreset of size $\core$
whose span contains an approximately optimal $q$-flat $\flat$.
If the distance of $\flat$ to any $p\in P$ is preserved under the projection, 
we can guarantee to preserve $\f_1^q(S,\rho)$ approximately as well. 
We ensure this preservation by the stronger property in Lemma~\ref{lem:dist_objects}
that for any $p\in P$, the distance to \emph{any} $q$-flat in the span of \emph{any} subset of $P$
of cardinality $\core$ is preserved.
Note that the number of such subspaces is bounded by $n^C$ and therefore polynomial in $n$.

Lemma~\ref{lem:dist_objects} in turn follows easily from a generalization of the Johnson-Lindenstrauss
lemma that we prove first: for an integer $c>0$, we show that a random projection to roughly $c\log(n)/\eps^2$
dimensions preserves for \emph{all} subset $S$ of $c$ points the distance between any two points in $\span{S}$.
While the proof of this subspace embedding result has been outlined in previous work~\cite{cw-low,nn-lower},
we are not aware of a formal proof of the statement.

\begin{lemma}\label{lem:sarlos_extended}
For $0<\eps<1$, a set $P\subset\reals^d$ of $n$ points, an integer $c\geq 0$, 
and $m\geq \lambda\cdot c\log(n)/\eps^2$ for a suitable constant $\lambda$, a random projection
$\proj$ satisfies with high probability that for any subset $S\subset P$
of cardinality $c$ and for any $u,v\in\span{S}$
\[(1-\eps)\|u-v\|\leq \|\proj(u)-\proj(v)\|\leq (1+\eps)\|u-v\|.\]
\end{lemma} 
\begin{proof}
The proof of Theorem~2.1 in Dasgupta and Gupta~\cite{dg-elementary} implies the following statement:
When projecting a unit vector in $\reals^d$ to a fixed $m = O(c\log n/\eps^2)$-dimensional subspace, the probability
that its squared length does not lie in $((1-\eps)m/d,(1+\eps)m/d)$ is at most
\[2\exp(-\frac{m\eps^2}{4})\leq 2\exp(-\frac{\lambda c\log n}{4})\leq n^{-8c}\]
for a suitable constant $\lambda$. As they argue, the same bound applies for a fixed unit vector
and a uniformly chosen $m$-dimensional subspace.

A result by Feige and Ofek~\cite{fo-spectral} (see also~\cite{ahk-fast}), translated in geometric terms, says that
by approximately preserving the pairwise squared distances between a set of at most $\exp(c\ln 18)$ sample points
belonging to an $c$-dimensional subspace, we can approximately preserve the squared length of all unit vectors in the subspace,
and thus all pairwise distances; see \cite[Proof of Cor.~11]{sarlos} for further explanations.
Hence, for a fixed subspace, we need to preserve $\exp(2c\ln 18)\leq \exp(6c)$ distances. 
Moreover, we want to preserve distances in $n^c$ many subspaces, yielding a total of $\exp(6c)n^c\leq n^{7c}$
distances to be preserved. By the union bound, choosing a $m$-dimensional subspace uniformly at random,
the probability of success is at least $1-\frac{n^{7c}}{n^{8c}}\geq 1-1/n^c$.
\end{proof}

The preservation of point-to-flat distances in low-dimensional subspaces
is a simple consequence:

\begin{lemma}\label{lem:dist_objects}
Let $0<\eps<1$, $P\subset\reals^d$ a set of $n$ points and $q < c$ positive integers.
With high probability, a random projection to an $O(c\log n/\eps^2)$-dimensional flat
satisfies for all subsets $S\subset P$ of cardinality $c$, all $q$-flats $Q\subset\span{S}$, and all $p\in P$ that
\[(1-\eps) d(p,Q) \leq d(\proj(p),\proj(Q))\leq (1+\eps)d(p,Q).\]
\end{lemma}
\begin{proof}
For any $p\in P$ and any $Q\subset\span{S}$, 
there exists a space with $c+1$ points that contains both $p$ and $Q$. 
Let $t\in Q$ be the point such that $d(p,Q)=\|p-t\|$. Applying Theorem~\ref{lem:sarlos_extended} for $c':=c+1$
immediately implies that $d(\proj(p),\proj(Q))\leq\|\proj(p)-\proj(t)\|\leq (1+\eps)d(p,Q)$.
The second inequality follows similarly, considering the point $t'\in Q$ that realizes $d(\proj(p),\proj(Q))$.
\end{proof}

We show our main theorem that random projections preserve $\f_1^q(S, \rho)$ for any $S\subseteq P$.

\begin{theorem}\label{thm:meta_theorem}
Let $0<\eps<1$, $P\subset\reals^d$ consist of $n$ points, $q\geq 0$ an integer and $\rho\in\Z_{\geq 1}\cup\{\infty\}$.
Then with high probability, for $m\geq \lambda\cdot C_\rho(q,\eps/2)\log(n)/\eps^2$ with a suitable constant $\lambda$,
a random projection $\proj$ satisfies for all subsets $S\subseteq P$
\[(1-\eps) \f_1^q(S, \rho) \leq \f_1^q(\proj(S), \rho) \leq (1+\eps)\f_1^q(S, \rho).\]
\end{theorem}
\begin{proof}
Let $S\subseteq P$ arbitrary. 
We start by showing the second inequality:
By the coreset property, there exists a subset $E\subset S$
of $C_\rho(q,\eps/2)$ points such that $\span{E}$ contains a $q$-flat $\flat$ that is an $\frac{\eps}{2}$-approximate solution.
For $\rho\neq\infty$, applying Lemma~\ref{lem:dist_objects} with $\eps'=\eps/3$, we get that
\begin{eqnarray*}
\f_1^q(\proj(S), \rho) &\leq& \left(\sum_{p\in S} d(\proj(p),\proj(\flat))^\rho\right)^{1/\rho} 
\leq \left(\sum_{p\in S} (1+\eps/3)^\rho d(p,\flat)^\rho\right)^{1/\rho} \\
&\leq& (1+\eps/3)(1+\eps/2) \f_1^q(S, \rho)\leq (1+\eps)\f_1^q(S, \rho),
\end{eqnarray*}
where we use $(1+\eps/3)(1+\eps/2)<1+\eps$ for $0\leq\eps\leq 1$.
For $\rho=\infty$, the proof for $\rho=1$ directly carries over.

For the first inequality, we apply the coreset property on the set $\proj(S)$: let $\proj(E')$ be a coreset for $\proj(S)$. Let $\Gflat$
denote the approximate solution in $\span{\proj(E')}$; it holds that $G=\proj(\flat')$
for some $q$-flat $\flat'$ in $\span{E'}$.
Using again Lemma~\ref{lem:dist_objects}, we have that
\begin{eqnarray*}
(1-\eps)\f_1^q(S, \rho)&\leq& (1-\frac{\eps}{2}) (1-\frac{\eps}{3})\left(\sum_{p\in S}d(p,\flat')^\rho\right)^{1/\rho}
\leq (1-\frac{\eps}{2})\left(\sum_{p\in S} d(\proj(p),\Gflat)^\rho\right)^{1/\rho}\\
& \leq& (1-\frac{\eps}{2})(1+\frac{\eps}{2})\f_1^q(\proj(S), \rho) \leq \f_1^q(\proj(S), \rho).
\end{eqnarray*}
Again, the case $\rho=\infty$ is analogue to $\rho=1$.
\end{proof}

Theorem~\ref{thm:meta_theorem} implies that $\f_k^q(P, \rho)$ is preserved for any~$k\geq 1$.

\begin{corollary}
\label{cor:k-q-flat_preserved}
With the notations of Theorem~\ref{thm:meta_theorem} and $k\geq 1$, a random projection $\proj$ satisfies with high probability
\[(1-\eps) \f_k^q(P, \rho) \leq \f_k^q(\proj(P), \rho) \leq (1+\eps)\f_k^q(P, \rho).\]
\end{corollary}
\begin{proof}
Let $\flatset=\{\flat_1,\ldots,\flat_k\}$ denote an optimal collection of $q$-flats, that is, for any $p\in P$, the closest flat in $\flatset$
has distance at most $\f_k^q(P, \rho)$. Let $P_i\subseteq P$ be the set of points closest to $\flat_i$, for $i=1,\ldots,k$. 
Note that $\flat_i$ is the optimal $q$-flat for $P_i$, in other words, it realizes $\f_1^q(P_i,\rho)$.%
\footnote{For $\rho=\infty$, this is not necessarily true for any optimal solution, but we can replace every $q$-flat with the optimal one wlog}
Using Theorem~\ref{thm:meta_theorem} on the subsets $P_i$, we get for $\rho<\infty$ that
\begin{eqnarray*}
\f_k^q(\proj(P), \rho)\leq \sum_{i=0}^{k} \f_1^q(\proj(P_i), \rho)
\leq \sum_{i=0}^{k}(1+\eps) \f_1^q(P_i, \rho)= (1+\eps)\f_k^q(P, \rho),
\end{eqnarray*}
proving the second inequality. The first part follows the same way considering an optimal $\flatset$ for $\proj(P)$.
The case $\rho=\infty$ is analogous, replacing all sums by max.
\end{proof}

\section{Coresets for Projective Clustering}
\label{sec:coresets}

Recall that $C_\rho(q,\eps)$ is defined as the coreset size for approximating the $L_\rho$-optimal $q$-flat,
in the sense that there exists a subset of $C_\rho(q,\eps)$ input points whose span contains
an $\eps$-approximate optimal $q$-flat. 

\subparagraph{The case of $0$-flats}
For a point set $P\subset\reals^d$, we consider the point that minimizes,
for a fixed $\rho\in[1,\infty]$, 
$$\delta(q):=\left(\sum_{p\in P} d(p,q)^\rho\right)^{1/\rho}$$
over all $q\in\reals^d$. 
We call the minimizer $o$ in $\reals^d$ the \emph{optimal center}
and note that $\delta(o)=f_1^0(P,\rho)$.
We call $o'$ an \emph{$\eps$-approximate center},
if $\delta(o')\leq (1+\eps)\delta(o)$.
Since Thm~\ref{thm:meta_theorem} only requires a bound on the coreset
and no method to compute it, we can free ourselves
from algorithmic considerations and concentrate on existential results.

A lower bound of $\Omega(1/\eps)$ can be derived easily by considering the standard simplex.
This has been done by B\u{a}doiu and Clarkson~\cite{bc-optimal} for the case $\rho=\infty$.

\begin{theorem}
\label{thm:coreset_lower_bound}
There exists a point set such that no subset of less than $1/(2\eps)$ points contains
an $\eps$-approximate center, i.e., $C_\rho(0,\eps)=\Omega(1/\eps)$.
\end{theorem}
\begin{proof}
Consider the standard $(n-1)$-simplex in $\reals^{n}$ spanned by the points $e_1,\ldots,e_n$.
The optimal center for any $\rho$ is the barycenter $o$ given by $(1/n,\ldots,1/n)$, and we have that
\[\delta(o)=n^{1/\rho}\sqrt{\frac{n-1}{n}}.\]
Choose a subset of size $c$, w.l.o.g. $e_1,\ldots,e_c$ and let $F$ be the affine subspace spanned by these points.
Let $o'$ denote the barycenter of the $(c-1)$-simplex $(e_1,\ldots,e_c)$. Now $o'$ is the point that
minimizes $\delta(\cdot)$ over $F$, since $o'$ is the orthogonal projection on $F$ of any $e_i$ with $i>c$. 
So, assuming that $F$ contains an $\eps$-approximate center, $o'$ must be an approximate center. On the other hand, $\delta(o')$
is easily computed by noting that $d(o',e_i)=\sqrt{(n-1)/n}$ for any $i\leq c$ and $d(o',e_i)=\sqrt{(n+1)/n}$ for any $i>c$. This yields
\begin{eqnarray*} 
\left( c \left(\sqrt{\frac{c-1}{c}}\right)^\rho + (n-c) \left(\sqrt{\frac{c+1}{c}}\right)^\rho\right)^{\frac{1}{\rho}} 
= \delta(o') \leq (1+\eps) \delta(o)=(1+\eps) n^{1/\rho}\sqrt{\frac{n-1}{n}}.
\end{eqnarray*}

Raising to the $\rho$-th power and dividing by $n$ yields that
\[\frac{c}{n} \left(\sqrt{\frac{c-1}{c}}\right)^\rho + \frac{n-c}{n} \left(\sqrt{\frac{c+1}{c}}\right)^\rho \leq(1+\eps)^\rho \left(\sqrt{\frac{n-1}{n}}\right)^{\rho}. \]
Because this bounds has to hold for every $n$, it must also hold in the limit for $n\rightarrow\infty$. That results in
\[\left(\sqrt{\frac{c+1}{c}}\right)^\rho \leq(1+\eps)^\rho, \]
and by solving for $c$ yields that $c\geq 1/(2\eps+\eps^2)\geq 1/(3\eps)$ for $\eps\leq 1$.
\end{proof}

The following result gives an almost tight upper bound for arbitrary $\rho$.
It follows directly from the techniques introduced
by Shyamalkumar and Varadarajan~\cite{sv-efficient} for the case of lines
through the origin. We give the proof for completeness.

\begin{theorem}
\label{thm:coreset_general_upper_bound}
For any (finite) point set $P\subset\reals^d$, there is a set $S\subset P$
of $O(1/\eps\log(1/\eps))$ points such that the subspace spanned by $S$ contains
an $\eps$-approximate center. In other words, $C_\rho(0,\eps)=O(1/\eps\log (1/\eps))$.
\end{theorem}
\begin{proof}
We define an iterative procedure which creates points $c_0,c_1,\ldots$
such that $c_i$ is in the subspace spanned by $i$ input points.
The initial point $c_0$ is chosen to be point closest to the optimal center.%
\footnote{Again, since we only care about existence, we can conveniently
assume that the center is known to us.}
If some $c_i$ is an $\eps$-approximate center, we are done. 
Otherwise, we chose a point $c_{i+1}$ that is significantly closer
to $o$. For that, let $s$ be the point
the maximizes
$$\frac{d(c_i,p)}{d(o,p)}$$
over all $p\in P$. By construction, $d(c_i,s)\geq (1+\eps)d(o,s)$
We choose $c_{i+1}$ as the point on the line segment $c_is$
that is closest to $o$.
It follows easily~\cite[Lemma 2.1]{sv-efficient}
that $d(c_{i+1},o)\leq (1-\eps/2)d(c_i,o)$. Combined with the triangle
inequality and the fact that $c_0$ is the closest point to $o$ in $P$,
this implies that for any $p\in p$:
\begin{eqnarray*}
d(c_k,p)\leq d(c_k,o)+d(o,p)
\leq (1-\eps/2)^k d(c_0,o)+d(o,p)
\leq (1+(1-\eps/2)^k) d(o,p).
\end{eqnarray*}
For $k=O(1/\eps\log(1/\eps))$, this means that $d(c_k,p)\leq (1+\eps) d(o,p)$ for all $p\in P$, 
which directly implies that $\delta(c_k)\leq (1+\eps)\delta(o)$.
\end{proof}

Smaller coresets exist for special cases: for $\rho=\infty$, a coreset of size $O(1/\eps)$
(in fact, of size $\lceil 1/\eps\rceil$) exists~\cite{bc-optimal}. 
We prove the same bound for the case $\rho=2$:

\begin{theorem}
\label{thm:k_means}
For $\rho=2$, there is a set of $O(1/\eps)$ points such that their subspace
contains an $\eps$-approximate center, i.e., $C_2(0,\eps)=O(1/\eps)$.
\end{theorem}
\begin{proof}
Writing $A$ for the matrix whose columns are the points in $P$
and $\Delta$ for the standard simplex with points $e_1,\ldots,e_n$, we can consider the function $g:\Delta\rightarrow\reals$ defined by
\begin{eqnarray*}
g(x)= \sum_{i=1}^n \| Ax - Ae_i\|^2
 = \sum_{i=1}^n (x-e_i)^T A^T A (x-e_i)
=x^T M x + x^T b + a
\end{eqnarray*}
with $M=nA^TA$, $b=-2\sum A^TAe_i$, and $a=\sum e_i^TA^TAe_i$. 
Therefore, $g$ is a quadratic function. 
We apply the Frank-Wolfe optimization~\cite{clarkson-coresets,jaggi-revisiting}
on the (convex) function $g$: this method starts in an arbitrary point $x_0$ in $P$
and improves the approximation quality in every step by moving towards the point in $P$ which the steepest descent.
The obtained sequence of iterates $x_0,x_1,\ldots$ 
converges to the (unique) minimum of $g$, and by construction, the iterate $x_i$ lies in the span of $i$ points of~$P$.

A crucial quantity in the convergence behavior of Frank-Wolfe is the quantity $C_g$: this is a scaled form
of the Bregman divergence of the function $g$, measuring the difference between $g(y)$ and the 
value at $y$ of the tangent plane of $g$ at $x$, for all pairs $x$ and $y$. Since $g$ is a quadratic function,
\cite[Sec. 4.3]{clarkson-coresets} asserts that $C_g\leq \diam(P)^2$. Writing $r$ for the radius of the MEB of $P$,
this implies $C_g\leq 4r^2$.

Slightly abusing notation, we let $o\in\Delta$ denote the point that minimizes $g$.
Using Theorem 2.3 from~\cite{clarkson-coresets}, after running the Franke-Wolfe optimization
for $k:=2\lceil1/\eps\rceil$ steps, we find an iterate $x_k$ on a $k$-simplex which satisfies
\[g(x_k)-g(o)\leq 4\eps C_f\leq 16\eps r^2\leq 16\eps g(o),\]
where the last inequality comes from the fact that with $p$ being the furthest point from $o$, it holds that
$g(o)\geq \|o-p\|^2\geq r^2$. Therefore, we have that $g(x_k)\leq (1+16\eps)g(o)$.
\end{proof}

\subparagraph{The case of general $q$}
The best known bounds for $C_\rho(q,\eps)$ with $q\geq 0$, are again due to Shyamalkumar and Varadarajan.
The aforementioned result for lines yields that $C_\rho(1,\eps)=O(1/\eps \log(1/\eps))$, the same bound as for $q=0$~\cite[Lemma 3.2]{sv-efficient}.
They use the line case in an inductive argument to show~\cite[Lemma 3.3]{sv-efficient}:

\begin{theorem}
\label{thm:sv-general-q-bound}
For $q\geq 1$, $C_\rho(q,\eps)=O(q^2/\eps \log(q/\eps))$.
\end{theorem}

A natural question is to ask about the tightness of the coreset bounds: for the point case $q=0$, we conjecture that coresets
of size $O(1/\eps)$ exist for any norm $\rho$ (currently, this is only established for $\rho\in\{2,\infty\})$.
For general $q$, an improved upper bound of $O(q/\eps)$ would yield a target dimension linear in $q$
in our dimension reduction result.

\section{Applications}
\label{sec:applications}

\subparagraph{Streaming algorithms for projective clustering}
We consider the projective clustering problem in a streaming context. In this setup, we do not return the cluster centers (the $q$-flats)
but only an $\eps$-approximation of $\f_k^q(P,\rho)$.
We let $S(n,d,q,k,\eps,\rho)$ be the space complexity for this problem. We assume that $n$, the size of the stream, is known in advance.

Set $m:=O((q+1)^2/\eps^3 \log n\log((q+1)/\eps))$. 
In the simplest variant, our streaming initially chooses a $d\times m$ projection matrix uniformly at random,
projects every point from the stream to $\reals^m$, and stores all points in a set $P'$. The algorithm, then uses an (offline)-algorithm to approximate $\f_k^q(P',\rho)$.
The total work space of this algorithm is $O(dm+nm+S)$, where $dm$ is the space required to store
the projection matrix, $nm$ is the size of $P'$, and $S$ is the space required to find approximate clustering of $P'$.
Using any approximation algorithm that computes using linear space,
we obtain a streaming algorithm to approximate  $\f_k^q(P,\rho)$ with a space complexity of $O((q+1)^2(n+d)/\eps^3 \log n\log((q+1)/\eps))$.
This is much smaller than the input size of $O(dn)$ and, for small $q$, not too far from the lower bound of $\Omega(n)$~\cite{as_soda10}.

In a similar fashion, our results can be used to speed up other streaming approaches:
Again, we choose initially a $d\times m$ projection matrix uniformly at random which is stored throughout the algorithm.
Furthermore, we maintain the workspace of a streaming algorithm 
that computes an approximation of the considered projective clustering problem in $m$ dimensions. 
When a new point $p \in \reals^d$ arrives, we compute its projection $\proj(p) \in \reals^m$ 
and treat this as an input to the $m$-dimensional streaming algorithm. 
We return the output value of the $m$-dimensional streaming algorithm as our result.
The correctness of the approach (with high probability) follows from Corollary~\ref{cor:k-q-flat_preserved}. The space complexity is $O(dm+S)$,
with $S$ the space complexity of the $m$-dimensional streaming algorithm.

\subparagraph{Approximate Projective Clustering.}
Our technique is also useful for the computation of approximate cluster centers:
For a set $P$ of $n$ points in $\reals^d$,  let $T(n,d,q,k,\eps,\rho)$ denote the time complexity to compute 
$k$ $q$-flats $\flatset$ that $\eps$-approximate the optimal solution, 
that is, $(\sum_{p\in P} (\dist(p, \flatset))^\rho)^{1/\rho}\leq (1+\eps)\f_k^q(P,\rho)$.
We design a new algorithm as follows:
Set $\eps':=\eps/5$.
First, we randomly project the input point set from $d$ to $m:=O(C_\rho(q,\eps')\log n/\eps'^2)$ dimensions. Let $P'$ be this set of projected points.
Then, we ($\eps'$-approximately) solve the same problem for $P'$ in $m$ dimensions, using some algorithm for this problem as a black box.
The computed solution clusters $P'$ in $k$ subsets of points that are closest to a particular $q$-flat in the solution.
We let $P^1,\ldots P^k$ be the pre-image of these $k$ clusters and assume wlog that $P^i\cap P^j=\emptyset$.
For each $P^i$, we compute an $\eps'$-approximation of the best fitting $q$-flat.
We return the collection of these $k$ $q$-flats as solution.
Correctness of this approach follows from Theorem~\ref{thm:meta_theorem}
and Corollary~\ref{cor:k-q-flat_preserved}.
As an example, we get the $k$-center problem by setting $\rho=\infty$ and $q=0$. Using the bounds $C_\infty(0,\eps)=2/\eps$, 
$T(n,d,0,k,\eps,\infty)=O(nd2^{k\log k/\eps})$ and $T(n,d,0,1,\eps,\infty)=O(\frac{nd}{\eps^2}+\frac{1}{\eps^5})$ from \cite{bc-smaller}, 
we get a running time of
\[O(n\log n 2^{k\log k/\eps}+\frac{dn\log n}{\eps^3}).\]

\subparagraph{Approximating \v{C}ech complexes}
A standard tool in capturing topological properties of point cloud data
is the \emph{\Cech complex}%
\footnote{For brevity, we omit a thorough introduction of the topological concepts used in this paragraph. See~\cite{eh-computational} for more details}. 
It is usually defined to be the \emph{nerve} of balls of some fixed radius $\alpha$
centered at the points from the sample $P$, and denoted as $\cech_\alpha(P)$. An equivalent definition is that
a $k$-simplex $\{p_0,\ldots,p_k\}$ is in $\cech_\alpha(P)$ if and only if
the radius of $\meb(p_0,\ldots,p_k)$ is at most~$\alpha$.

The downside of \Cech complexes is the size: Their $d$-skeleton can consist of up to $O(n^{d+1})$ simplices. 
Recent work suggests to work instead with an approximation of the \Cech complex~\cite{ks-cech}
(or of the closely related Vietoris-Rips complex~\cite{sheehy-linear}\,\cite{dfw-graph}).
``Approximation'' in this context means that the persistence diagrams of the modules induced by the \Cech filtration and by the approximate filtration are close to each other.
Theorem~\ref{thm:meta_theorem} for $q=0$, $k=1$ and $\rho=\infty$ implies that the radius of MEBs is preserved for any subset.
That implies immediately that \Cech complexes can be approximated by \Cech complexes in lower dimensions.

\begin{proposition}
\label{dim-reduction-cech}
For $0<\eps\leq\frac{c-1}{c}<1$ with $c>1$ and arbitrary constant, 
a set $P\subset\reals^d$ of $n$ points, and $m=\Theta(\log(n)/\eps^3)$,
a random projection $\proj:\reals^d\rightarrow \reals^m$ satisfies with high probability that
\[\cech_{(1-c\eps)\alpha}(P)\subseteq\cech_\alpha(\proj(P))\subseteq\cech_{(1+c\eps)\alpha}(P).\]
\end{proposition}

An interesting consequence of this statement is that a \Cech complex cannot have
any significantly persistent features in dimensions higher than $m$. Independently from our work, Sheehy~\cite{sheehy-persistent} 
recently showed a slightly stronger result, projecting to $\Theta(\log(n)/\eps^2)$ dimensions.

\subparagraph{Acknowledgments.} The authors thank the anonymous referees of an earlier version of this paper
whose comments have led to significant simplifications and improvements of our results. The first author 
acknowledges support by the Max Planck Center for Visual Computing and Communication.
The second author acknowledges support by NSF through grant CCF-1464276.

\end{document}